\documentclass[12pt, draftclsnofoot, onecolumn]{IEEEtran} 
\IEEEoverridecommandlockouts
\usepackage{algpseudocode}
\usepackage{algorithm}
\usepackage{color}
\usepackage{bm}
\usepackage{amsfonts}
\usepackage{amssymb}
\usepackage{amscd}
\usepackage{graphics}
\usepackage[dvips]{graphicx}
\usepackage{epsfig}
\usepackage{subfigure}
\usepackage{caption}
\usepackage[cmex10]{amsmath}
\usepackage{array}
\usepackage{eqparbox}
\usepackage{flushend}
\usepackage{hyperref}
\usepackage{fancyhdr}
\usepackage{titling}
\usepackage{enumerate}
\usepackage{stackrel}
\usepackage{pdflscape}
\usepackage{longtable}
\usepackage{lineno}
\usepackage{titling}
\usepackage{dsfont}

\newcommand{\ie}{{\em i.e.}}

\newcommand{\iid}{i.i.d.}
\newcommand{\apriori}{{\em a priori}}

\newcommand{\secref}[1]{Section~\ref{#1}}
\newcommand{\figref}[1]{Fig.~\ref{#1}}

\newcommand{\thrmref}[1]{Theorem~\ref{#1}}

\newcommand{\appref}[1]{Appendix~\ref{#1}}
\newcommand{\defref}[1]{Definition~\ref{#1}}

\newtheorem{thrm}{\textbf{Theorem}}

\newtheorem{defn}{\textbf{Definition}}

\newcommand{\abs}[1]{\left\vert#1\right\vert}
\newcommand{\norm}[1]{\Vert#1\Vert}

\makeatletter
\def\blfootnote{\xdef\@thefnmark{}\@footnotetext}
\makeatother

\newenvironment{proof}[1][Proof]{\begin{trivlist}
\item[\hskip \labelsep {\bfseries #1}]}{\end{trivlist}}

\newcommand{\qed}{\nobreak \ifvmode \relax \else
      \ifdim\lastskip<1.5em \hskip-\lastskip
      \hskip1.5em plus0em minus0.5em \fi \nobreak
      \vrule height0.75em width0.5em depth0.25em\fi}

\def\BibTeX{{\rm B\kern-.05em{\sc i\kern-.025em b}\kern-.08em
    t\kern-.1667em\lower.7ex\hbox{E}\kern-.125emX}}

\usepackage{newlfont}
\hypersetup{pageanchor=false}
\begin{document}
\title{Caching With Time-Varying Popularity Profiles: A Learning-Theoretic Perspective} 
\author{\IEEEauthorblockN{B. N. Bharath, K. G. Nagananda, D. G\"{u}nd\"{u}z, and H. Vincent Poor}\thanks{B. N. Bharath is with Indian Institute of Technology, Dharwad, INDIA, E-mail: \texttt{bharathbn@iitdh.ac.in}. K. G. Nagananda was with PES University, INDIA, E-mail: \texttt{kgnagananda@pes.edu}. D. G\"{u}nd\"{u}z is with Imperial College London, UK, E-mail: \texttt{d.gunduz@imperial.ac.uk}. H. Vincent Poor is with Princeton University, New Jersey, USA, E-mail: \texttt{poor@princeton.edu}. This work was supported in part by the U.S. National Science Foundation under Grants CCF-1420575 and CNS-1456793, the European Research Council (ERC) under Starting Grant BEACON (agreement 677854), DST/INT/UK/P-129/2016 and the Startup Grant from IIT, Dharwad. }}

\date{}
\maketitle

\vspace{-0.75in}
\begin{abstract}
Content caching at the small-cell base stations (sBSs) in a heterogeneous wireless network is considered. A cost function is proposed that captures the backhaul link load called the ``offloading loss'', which measures the fraction of the requested files that are not available in the sBS caches. As opposed to the previous approaches that consider time-invariant and perfectly known popularity profile, caching with non-stationary and statistically dependent popularity profiles (assumed unknown, and hence, estimated) is studied from a learning-theoretic perspective. A probably approximately correct result is derived, which presents a high probability bound on the offloading loss difference, {\ie}, the error between the estimated and the optimal offloading loss. The difference is a function of the Rademacher complexity, the $\beta-$mixing coefficient, the number of time slots, and a measure of discrepancy between the estimated and true popularity profiles. A cache update algorithm is proposed, and simulation results are presented to show its superiority over periodic updates. The performance analyses for Bernoulli and Poisson request models are also presented. 
\end{abstract}


\section{Introduction} \label{sec:intorduction}
Wireless data traffic is growing at an unprecedented rate, exacerbating the demand for improved design strategies for the next generation wireless infrastructure \cite{Furuskar2015}. Deployment of small base stations (sBSs) to offload wireless data from a macro base station (BS) can have the potential to not only improve the network performance during peak data traffic periods, but also to integrate existing WiFi and cellular technologies in an efficient manner  \cite{Bennis2013}, \cite{Chou2014}. A potential drawback of the small-cell infrastructure to offload wireless data from a macro BS is that the backhaul link-capacity required to support the peak data traffic can be alarmingly high, necessitating complex and expensive solutions to ensure high throughput and performance during peak traffic periods. Caching can reduce the peak backhaul load by storing popular contents in local cache memories located at the sBSs \cite{Niesen2012}. Benefits of coded caching across sBSs is shown in \cite{Golrezaei2012} and \cite{Xu2017}, while in \cite{Bacstuug2015a} caching is analyzed for networks modeled using independent Poisson point processes (PPPs). The performance of TCP is shown to improve with the help of caching in \cite{Hu2003}, while caching-based content-centric networking, and an information-centric architecture for energy-efficient content distribution are proposed in \cite{Wang2014} and \cite{Fang2014}, respectively. Results on caching video files and their benefits are presented in \cite{Pedarsani2014} - \nocite{Golrezaei2014}\cite{Li2016}, while the advantages of data caching and content distribution in device-to-device (D2D) communications are studied in \cite{Ji2016} - \nocite{Zhang2016}\cite{Chen2017a}. In \cite{Gregori2016}, proactive caching is shown to increase the energy efficiency of D2D communications, while the advantages of caching on mobile social networks is reported in \cite{Wu2016}.

Most papers in the literature assume {\apriori} knowledge of the popularity profile of the cached contents, which is unreasonable in practical scenarios. This assumption is relaxed in \cite{Blasco2014} - \nocite{Blasco2014a}\cite{Bacstuug2015}, and various learning-based approaches are proposed to estimate the popularity profile, and theoretical analyses have been carried out to study the implications of learning the popularity profile and user preferences on the performance \cite{Golrezaei2013} - \nocite{Tatar2014}\nocite{Bharath2016}\nocite{Song2017}\cite{Chen2016}. However, these works assume that the popularity profile is stationary and statistically independent across time. In practice, there are many applications (for example, video on demand) in which the popularity profile of cached contents is a function of time \cite{Cha2009} - \nocite{Szabo2010}\cite{Kim2017}. Motivated by these applications and the growing significance of caching in improving the quality of service for end-users during peak traffic periods, we analyze the performance of a random caching strategy for a \emph{non-stationary} popularity profile, which may have statistical dependence across time.

A heterogenous network in which the users, BSs, and sBSs are distributed according to independent PPPs is considered. The sBSs employ a random caching strategy. A protocol model for communication is proposed, and a cost function, which captures the backhaul link overhead called the ``offloading loss'', is considered. The offloading loss at time $t$, which depends on the popularity profile, is denoted by $\mathcal{T}(t)$. Our goal is to obtain risk bounds on this offloading loss when the popularity profile is time-varying and unknown. Under a certain request model (see \texttt{Assumption $1$}), the BS first estimates the popularity profile based on the requests observed during the first $t$ slots. It then chooses the caching probabilities $\pi \triangleq (\pi_1,\pi_2,\ldots, \pi_N)$, where $N$ is the number of popular content items that can be cached, in order to minimize its offloading loss $\hat{\mathcal{T}}(t)$, based on the estimated popularity profile. sBSs in the coverage area of the BS use this optimal caching policy to store content items in their caches. Since the popularity profile is time-varying, it becomes necessary to frequently refresh the caches, say after every $T$ time slots, albeit at an additional cost. Thus, it is important to investigate the minimum periodicity $T$ of cache updates that guarantees a desired offloading loss.

In this paper, we derive probably approximately correct (PAC) type guarantees on the \emph{offloading loss difference} $\Delta_{\mathcal{T}}(t,T)$, which is defined as the difference between the offloading loss incurred by using the outdated caching policy obtained by optimizing $\hat{\mathcal{T}} (t)$ at time $t+T$, and the optimal offloading loss at time $t+T$. We show that $\Delta_{\mathcal{T}}(t,T) < \epsilon$ with a probability of at least $1-\delta$ for any $\delta > \zeta$ and $\epsilon >  0$, where $\zeta$ is a function of the $\beta$-mixing coefficient, the number of content items $N$, and the user density. The $\beta$-mixing coefficient is a measure of the statistical dependency of the time-varying popularity profiles. If the popularity profile process is ``sufficiently'' mixing, {\ie}, if the process becomes almost independent after a sufficiently long time, and if the user density is very high, then the desired $\epsilon$ can be achieved for negligibly small $\delta > 0$. In particular, to achieve a fixed probability $\delta > \zeta$, we require the error $\epsilon$ to be a function of $N$, the rate of change of the popularity profile, and the Rademacher complexity, which is a measure of the difficulty in estimating the offloading loss. 

The following are the main findings of this paper: (1) the error $\epsilon$ increases with $N$; (2) the desired error $\epsilon$ can be achieved with higher probability (i.e., $\zeta$ becomes smaller) for a larger user density, thus improving the caching performance, since higher user density results in more user-requests, allowing a better estimate of the popularity profile; (3) the higher the correlation of the popularity profile across time (defined in terms of the $\beta-$mixing coefficient), the longer the waiting time $t$ to achieve a target error level $\epsilon$ with probability $1-\delta$; (4) the error $\epsilon$ is a function of the rate of change of the popularity profile, and hence, the cache refresh period $T$. Thus, outdated cache contents lead to a larger error for a given $\delta$, and a rapidly varying popularity profile requires more frequent updates to achieve the desired error performance; (5) a higher Rademacher complexity results in poorer error performance; and (6) when the user requests are independent and identically distributed ({\iid}), the error performance is better compared to non-stationary and statistically dependent requests. For stationary popularity profiles and large $t$, frequent cache-updates are not necessary to achieve the desired performance. Finally, motivated by our theoretical bounds, we present an algorithm which updates the cache contents only if the discrepancy that captures the rate at which the popularity profile is changing, is large. We demonstrate the benefits of using the proposed cache update policy compared to periodic cache updates through simulations. To the best of our knowledge, this is the first time random caching is studied with non-stationary, statistically dependent, and unknown popularity profiles from a learning theory perspective. The initial results of this work can be found in \cite{Bharath2017}. 

The remainder of the paper is organized as follows. In \secref{sec:sys_model}, we present the system model and introduce the notation. The problem statement is introduced in \secref{sec:problem_statement}, while the main results are presented in \secref{sec:mainresult_1}. Performance analyses for Bernoulli and Poisson request models are analyzed in \secref{sec:bernoulli_poisson}. Numerical results are presented in \secref{sec:numerical_results}. Concluding remarks are provided in \secref{sec:conclusion}.

\section{System Model} \label{sec:sys_model}
A heterogenous cellular network is considered in which the users, BSs and sBSs are spatially distributed according to independent PPPs with densities $\lambda_u$, $\lambda_b$ and $\lambda_s$, respectively \cite{Baccelli1997}. The sets of users, BSs and sBSs are denoted by $\Phi_u \subseteq \mathbb{R}^2$, $\Phi_b \subseteq \mathbb{R}^2$, and $\Phi_s \subseteq \mathbb{R}^2$, respectively. Each user requests a content item (i.e., {\it file}) from the library $\mathcal{F} \triangleq \{f_1, \ldots, f_N\}$ of $N$ files, each of size $B$ bits, from its neighboring sBSs. The requests are assumed to be statistically independent across users. However, the requests from each user are assumed to be \emph{non-stationary} and statistically \emph{dependent} across time. We assume that the size of the cache at each sBS is at most $M$ files. 
The problem considered in this paper is that of caching relevant ``popular" files at the sBSs, wherein, depending on the availability of the file in the local cache, the file requested by a user will be served directly by one of its neighboring sBSs. 
In order to access cached content items, a user $u \in \Phi_u$ identifies and communicates with a set of neighboring sBSs employing the following protocol: sBS $s$ located at $x_s \in \Phi_s$ communicates with user $u$ located at $x_u  \in \Phi_u$ if $\norm{x_u - x_s} < \gamma$, for some $\gamma > 0$. This condition determines the communication radius. In this protocol, we ignore the interference from other users in the network. The set of neighbors of user $u$ located at $x_u$ is denoted by $\mathcal{N}_u \triangleq \{y \in \Phi_s: \norm{y - x_u} < \gamma\}$. The caching policy will depend on the distribution of the requests from the users, which is assumed to be unknown, and should be estimated. In the next subsection, we present a stochastic process modeling the requests from the users, and devise a method for estimating its distribution.  

\subsection{User Request Model}
\begin{figure}[h!]
\begin{center}
{\includegraphics[height=5cm,width=12.0cm]{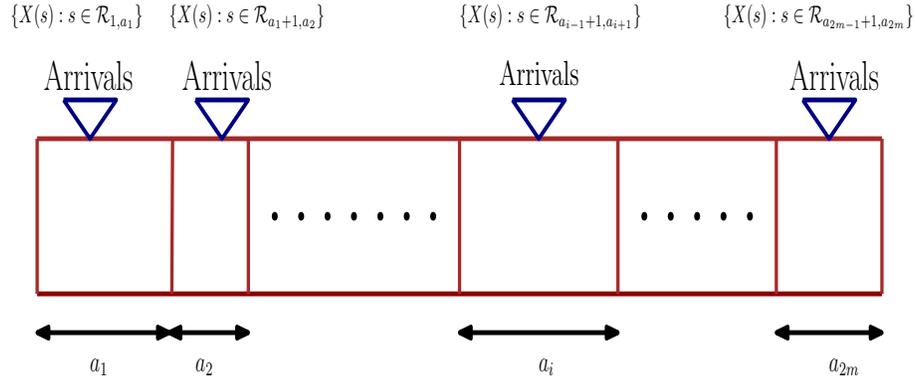}}
\caption{A time period consisting of $t$ time slots, each of duration $\Delta$, is divided into $2m$ blocks, where the $i^{\text{th}}$ block consists of $a_i$ slots, and $t = \sum_{i=1}^{2m} a_i$.} 
\label{fig:figslot_timediv}
\end{center}
\end{figure}
Let the stochastic process $X_v(\tau) \in  \{1,2,\ldots,N\}$ denote the index of the requested file by user $v \in \Phi_u$ at time $\tau \in \mathbb{R}$. For example, each user can maintain an independent local Poisson clock, and makes a request whenever the local clock ticks. For any two users $v,w \in \Phi_u$, the request processes $X_v(\tau)$ and $X_w(\tau)$ are independent. For the ease of analysis, let us divide the time into slots of size $\Delta > 0$ each. Further, for each $v \in \Phi_u$, $\{X_v(\tau), \tau \in \mathbb{R}\}$ is a non-stationary and statistically dependent stochastic process across time slots, but the process $X_v(\tau)$ within each time slot (i.e., $\tau \in [i \Delta, (i+1) \Delta)$, $i=1,2,\ldots$) is assumed to be stationary. Further, we assume that there is a ``typical" BS at the origin with a coverage radius of $R > 0$. The BS estimates the popularity of the content items based on the requests it receives. Essentially, at a given time slot $t$, the BS collects requests (for $t$ time slots) from all the users in the BS's coverage area to estimate the popularity profile of the requested files. Let $n_u \sim \text{Poiss}(\pi \lambda_u R^2)$ denote the number of users in its coverage area. The random arrival instants of the requests from different users are assumed to satisfy the following assumption. 

\textbf{Assumption $1$:} There exist constants $0 \leq \alpha_{\texttt{min}} \leq \alpha_{\texttt{max}} \leq 1$ such that for any random $n_u = n \geq 1$ users in the coverage area of the BS, the number of requests in $a \in \mathbb{N}$ time slots, denoted by $r_a \in \mathbb{N}$, satisfies $\Pr\{\alpha_{\texttt{min}} n a \leq r_a \leq \alpha_{\texttt{max}} n a \left\vert \right. n_u = n\} > \zeta_{a,n}$ for some $\zeta_{a,n} > 0$. \label{def:request_model}

It turns out that the results based on the above assumption can be used to derive performance guarantees when the arrival process is a homogenous Poisson point process (see Sec.~\ref{sec:bernoulli_poisson}). Further, we assume that the request instants and the number of requests within a time slot are independent of the files requested. The set of request instants at which the requests from all the users in the coverage area of the BS arrive within the $i^{\text{th}}$ time slot is denoted by $\mathcal{R}_{i}$. Let $X(\tau) \triangleq \bigcup_{v \in \Phi_u \bigcap \norm{v}_2 \leq R} \{X_v(\tau)\}$ denote the set of requests from all the users in the coverage area of the BS at time $\tau \in \mathbb{R}$. Note that if two or more users request for the same file at time $\tau \in \mathbb{R}$, then it is counted as the same index due to the union in the definition of $X(\tau)$. However, this event does not occur almost surely. The set of requests from all the users in time slots $t_1$ to $t_2$ is denoted by $X_{t_1, t_2} \triangleq \{{X}(\tau): \tau \in \mathcal{R}_{t_1,t_2}\}$, where $\mathcal{R}_{t_1,t_2} \triangleq \bigcup_{i=t_1}^{t_2} \mathcal{R}_{i}$ (see Fig. \ref{fig:figslot_timediv}). After receiving requests $X_{1, t}$ within first $t$ time slots, the BS computes the empirical estimate of the popularity profile, {\ie}, the probability of the $i^{\text{th}}$ file being requested is estimated as follows:
\begin{eqnarray}
\hat{p}_{i,t} = \frac{1}{r_t} \sum_{s \in \mathcal{R}_{1,t}} \mathds{1}\{X(s) = i\},~ i=1, \ldots, N,
\label{eq:estimation_popularity}
\end{eqnarray}
where $r_t \triangleq \abs{\mathcal{R}_{1,t}}$ is the total number of requests in the first $t$ slots, and the indicator function $\mathds{1}\{X(s) = i\}$ is one when the event $\{X(s) = i\}$ occurs, zero otherwise. The accuracy of the estimate $\hat{\mathcal{P}}^{(t)} \triangleq \{\hat{p}_{i,t}: i = 1,2,\ldots,N\}$ depends on (i) the number of available samples, which in turn is related to the number of users in the coverage area of the BS, (ii) the number of requests per user, and (iii) the behavior of the process $X(s)$. The estimate in \eqref{eq:estimation_popularity} is valid only when there is a positive number of user requests, which is guaranteed by \texttt{Assumption $1$} above. In the next section, we present the performance measure for the above model, and state the main problem addressed in the paper. 

\section{Problem Statement}\label{sec:problem_statement}
We consider a typical user located at the origin denoted by $o \in \Phi_u$. At time slot $t \in \mathbb{N}$, the ``offloading loss'' is defined as
\begin{eqnarray}
\mathcal{T}(\Pi^{(t)}, \mathcal{P}^{(t)}, X_{1,t-1})  \triangleq \frac{B}{R_0}\Pr\left\{f_o \notin \mathcal{N}_u \left \vert \right. X_{1,t-1} \right\},
\label{eq:metric}
\end{eqnarray}
where $\Pi^{(t)}$ denotes the caching policy, $\mathcal{P}^{(t)} \triangleq \{p_1(t),p_2(t),\ldots,p_N(t)\}$ is the popularity profile in slot $t$, $R_0$ and $\frac{B}{R_0}$ denote the rate supported by the BS and the time overhead incurred in transmitting the file from the BS to the user, respectively, and  $f_o$ denotes the file requested by the typical user in the $t$-th slot. In \eqref{eq:metric}, with a slight abuse of notation, $f_0 \notin \mathcal{N}_u$ denotes the event that the requested file $f_0$ is not present in the caches of the neighboring sBSs. The offloading loss is the scaled probability of the content requested by user $o$ not being cached by any of the sBSs within its communication range conditioned on the requests received by the BS until the beginning of time slot $t$, {\ie}, $X_{1,t-1}$. We employ the following random caching strategy, which enables us to derive a closed form expression for the offloading loss at time $t$.

\textbf{Random caching strategy:} At time $t$ (determined by the BS), each sBS $s \in \Phi_s$ caches content items in an {\iid} fashion by generating $M$ indices distributed according to $\Pi^{(t)} \triangleq \left\{\pi_i(t):  \sum_{i=1}^N \pi_i(t) = 1, \right\}$ (see \cite{Ji2013}).

We seek to solve the following optimization problem:
\begin{eqnarray}
&\min\limits_{\Pi^{(\tau)} \in \mathcal{P}_{\pi}:\tau \in \mathbb{N}}&\limsup_{t \rightarrow \infty} \frac{1}{t} \sum_{\tau = 1}^t\mathcal{T}(\Pi^{(\tau)},\mathcal{P}^{(\tau)}, X_{1,\tau-1}),
\label{eq:main_opt_problem}
\end{eqnarray}
where $\mathcal{P}_{\pi}$ denotes the $N-$dimensional probability simplex. An expression for $\mathcal{T}(\Pi^{(t)},\mathcal{P}^{(t)}, X_{1,t-1}) $ is given in the following theorem, whose proof can obtained by replacing $p_i$ by $p_{X,i}(t)$ in the proof of Theorem $1$ found in \cite[Appendix A]{Bharath2016}.

\begin{thrm}
The average offloading loss at time $t$ for the random caching strategy $\Pi^{(t)}$ is given by
\begin{eqnarray}
\mathcal{T}(\Pi^{(t)},\mathcal{P}^{(t)}, X_{1,t-1})  =   \sum_{i=1}^N g(\pi_i(t)) p_{X,i}(t),
\label{eq:mean_througput_expression}
\end{eqnarray}
where $p_{X,i}(t)  \triangleq \Pr \{f_i \text{ requested by $o$ in slot $t$} | X_{1,t-1}\}$, and $g(\pi_i(t)) \triangleq \frac{B}{R_0}\exp\{-\lambda_u \pi \gamma^2[1-(1-\pi_i(t))^M]\}$.
\label{thrm:thm_mean_throughput}
\end{thrm}

Even assuming that the conditional probabilities $p_{X,i}(t)$ are perfectly known, the complexity involved in solving \eqref{eq:main_opt_problem} can be high owing to the fact that the caching policy at time $t$ depends on $X_{1,t}$, which grows with $t$. In practice, the conditional probability $\Pr \{f_i \text{ requested } | X_{1,t-1}\}$ is unknown, and has to be estimated. Also, the BS may not have enough samples to compute a reasonably good estimate of this conditional probability. Hence, it is reasonable to consider the unconditional probability in the definition of the offloading loss. Thus, one can minimize the offloading loss $\mathcal{T}(\Pi^{(t)},\mathcal{P}^{(t)}) \triangleq \left[ \sum_{i=1}^N g(\pi_i(t)) p_{i}(t) \right]$, where $p_i(t)$ is the probability of the $i^{\text{th}}$ file being requested at time $t$. However, the $p_i(t)$'s are unknown; and hence, an estimate of the popularity profile needs to be used in place of $\mathcal{P}^{(t)}$. More precisely, at time $t$, let $\hat{\Pi}^{*}_t$ denote the caching policy obtained using an estimate $\hat{\mathcal{P}}^{(t)}$; that is,
\begin{eqnarray}
\hat{\Pi}^{*}_t = \arg \min_{\Pi^{(t)} \in \mathcal{P}_{\pi}} ~~ {\mathcal{T}}(\Pi^{(t)},\hat{\mathcal{P}}^{(t)}).
\label{eq:opt_problem_emperical}
\end{eqnarray}
Suppose that the cache contents chosen by the optimal caching policy at time $t$ will be used to satisfy user demands over the period $(t, t+T]$. Let us consider the offloading loss in using $\hat{\Pi}^{*}_t$ at a later time, say at time $t+T$. The offloading loss at time $t + T$ is given by $\hat{\mathcal{T}}^*(t + T) \triangleq \mathcal{T}(\hat{\Pi}^{*}_t,{\mathcal{P}}^{(t + T)})$. Further, let ${\Pi}^{*}_{t + T}$ denote the optimal caching policy at time $t + T$ using perfect knowledge of the popularity profile $\mathcal{P}^{(t+T)}$; that is,
\begin{eqnarray}
{\Pi}^{*}_{t + T} = \arg \min_{\Pi^{(t + T)} \in \mathcal{P}_{\pi}} ~ {\mathcal{T}}(\Pi^{(t+T)},{\mathcal{P}}^{(t + T)}),
\label{eq:opt_problem}
\end{eqnarray}
with the corresponding offloading loss ${\mathcal{T}}^*(t + T) \triangleq \mathcal{T}({\Pi}^{*}_{t + T},{\mathcal{P}}^{(t + T)})$. Similar to \cite{Bharath2016}, the central theme of this paper is the analysis of the \emph{offloading loss gap}, $\Delta_{\mathcal{T}}(t,T) \triangleq \hat{\mathcal{T}}{^*}(t + T) - \mathcal{T}{^*}(t + T)$. For example, if $\Delta_{\mathcal{T}}(t,T)$ is small, then each term in \eqref{eq:main_opt_problem} is small, which results in a small average offloading loss. This approach is central to the analyses of prediction problems involving non-stationary stochastic processes \cite{Kuznetsov2014}. 

The number of requests in any given slot and the requested file index are independent. For example, if the arrivals are Poisson, then the number of requests in any two disjoint intervals are independent. However, the files requested across time are correlated. This assumption is reasonable when the popularity depends on, for example, the files that are trending due to their popularity elsewhere, while a user's decision to browse is independent of the popularity. The unconditional probability does not lead to the independence of the files requested in any slot $t$ from the files requested in future slots. Moreover, an estimate of the  popularity profile at time slot $t$ depends on the past requests. However, for future work we aim to investigate generalization bounds retaining the conditioning on the past requests, which makes the offloading loss $\mathcal{T}(\Pi^{(t)}, \mathcal{P}^{(t)}, X_{1,t-1})  \triangleq \frac{B}{R_0}\Pr\left\{f_o \notin \mathcal{N}_u \left \vert \right. X_{1,t-1} \right\}$ at any given slot $t$ random. 

\section{Main Results} \label{sec:mainresult_1}
We study risk bounds on the offloading loss difference, $\Delta_{\mathcal{T}}(t,T)$, when the popularity profile is non-stationary. Essentially, for any $\epsilon > 0$, we seek to identify a risk bound $\delta > 0$, such that
\begin{eqnarray} \label{eq:prob_dff}
\Pr\left\{\hat{\mathcal{T}}{^\ast}(t + T)  -  \mathcal{T}{^\ast}(t + T) > \epsilon\right\} < \delta.
\end{eqnarray}
First, we relate \eqref{eq:prob_dff} to an expression in terms of the estimation error in the following theorem.
\begin{thrm}
 For the estimate of the popularity profile in \eqref{eq:estimation_popularity}, the following bound holds:
\begin{eqnarray}
\nonumber \Pr\left\{\hat{\mathcal{T}}{^*}(t + T)  -  \mathcal{T}{^*}(t + T) > \epsilon \right\} \leq 2\Pr\left\{\mathcal{A}_{T}(X_{1,t})  > \epsilon \right\},
\end{eqnarray}
where $\mathcal{A}_{T}(X_{1,t}) \triangleq \sup_{\Pi \in \mathcal{P}_\pi} \abs{\sum_{i=1}^N g(\pi_i) (\hat{p}_{i,t} - p_{i,t+T})}$, and $g(\pi_i)$ is defined in \thrmref{thrm:thm_mean_throughput}.
\label{thm:machine_learning}
\end{thrm}
\begin{proof} 
See \appref{app:proof_ml}.
\end{proof}

The term $\Pr\left\{\mathcal{A}_{T}(X_{1,t})  > \epsilon \right\}$ can be bounded as follows:
\begin{eqnarray}
\nonumber \Pr\left\{\mathcal{A}_{T}(X_{1,t})  > \epsilon \right\} &=& \sum_{j=0}^\infty \Pr\left\{\mathcal{A}_{T}(X_{1,t})  > \epsilon \left \vert \right. n_u = j\right\} \Pr\{n_u = j\} \\
\nonumber &\leq& \Pr\left\{n_u = 0 \right\} +   \sum_{j=1}^\infty \Pr\left\{\mathcal{A}_{T}(X_{1,t})  > \epsilon \left \vert \right. n_u = j\right\} \Pr\{n_u = j\}  \\
 &=& \exp\left\{-\lambda_u \pi R^2\right\}  + \sum_{j=1}^\infty \Pr\left\{\mathcal{A}_{T}(X_{1,t})  > \epsilon \left \vert \right. n_u = j\right\} \Pr\{n_u = j\}
\label{eq:first_bound}.
\end{eqnarray}
We next derive an upper bound on $\Pr \left\{ \mathcal{A}_{T}(X_{1,t})  > \epsilon | n_u = j\right\}$. The term $\mathcal{A}_{T}(X_{1,t})$ depends on $\hat{p}_{i,t}$, which involves the sum of non-stationary random variables which are possibly correlated across time. In order to apply the standard large deviation bounds, we must convert the sum of non-stationary dependent random variables to a sum of blocks of independent random vectors through a coupling argument, which is explained next.

For a given stochastic process $X_{1,\infty}$, and $s \in \mathbb{N}$, let $\mathbb{P}_{\tau,\tau + s}(\star)$ and $\mathbb{P}_{1 \rightarrow \tau}(\star) \otimes \mathbb{P}_{\tau + s \rightarrow \infty}(\star )$ denote the joint and product distributions of the stochastic processes $X_{1,\tau}$ and $X_{\tau+s, \infty}$, respectively. If $X_{1,\tau}$ and $X_{\tau + s, \infty}$ are independent, then $\norm{\mathbb{P}_{\tau,\tau + s}(\star ) - \mathbb{P}_{1 \rightarrow \tau}(\star ) \otimes \mathbb{P}_{\tau + s \rightarrow \infty}(\star)}_{\texttt{TV}} = 0$, where $\norm{\star}_{\texttt{TV}}$ denotes the total variational norm. Thus, for a given $s$, this difference, maximized over all $1 \leq \tau \leq \infty$ is a natural measure of the dependency between $X_{1,\tau}$ and $X_{\tau + s, \infty}$. This is commonly referred to as the $\beta-$mixing coefficient, and for $s \in \mathbb{N}$, it is given by
\begin{equation}
\beta(s) \triangleq \sup_{1 \leq \tau \leq \infty} \norm{\mathbb{P}_{\tau,\tau + s}(\star ) - \mathbb{P}_{1 \rightarrow \tau}(\star)\\ \otimes \mathbb{P}_{\tau + s \rightarrow \infty}(\star) }_{\textit{TV}}.
\end{equation}
A stochastic process is said to be $\beta$-mixing if $\beta(s) \rightarrow 0$ as $s \rightarrow \infty$. For a given stochastic process that is $\beta$-mixing, two well-separated sequences of the process are approximately independent, where the approximation error is given by $\beta(s)$. Thus, we assume that the request process $X(t)$ is a $\beta$-mixing stochastic process, i.e., $\beta(s) \rightarrow 0$ as $s \rightarrow \infty$. 

We now provide the details of the coupling argument, through which the dependent stochastic process is replaced by independent blocks of random variables. This will facilitate the use of a concentration inequality; in particular, McDiarmid's inequality. 
Fix $m \in \mathbb{N}$, and consider $2m$ consecutive blocks, where the block $i$, $i \in \{1,2,\ldots,2m\}$, consists of $a_i$ time slots, and $t \triangleq \sum_{j=1}^{2m} a_j$ is the total number of time slots (see Fig. \ref{fig:figslot_timediv}). Let $a_0 \triangleq 0$. Consider the time instants at which the requests arrive corresponding to odd and even blocks defined as $\mathbb{T}^{(t)}_o \triangleq \bigcup_{j:j = 0,2,4,\ldots,2(m-1)} \mathcal{R}_{a_j+1,a_{j+1} }$ and $\mathbb{T}^{(t)}_e \triangleq \bigcup_{j:j = 1,3,5,\ldots,2m-1} \mathcal{R}_{a_j+1,a_{j+1} }$, respectively. Thus, the requests corresponding to the odd and even blocks are given by $X^e_{1,t} \triangleq \{X(s): s \in \mathbb{T}^{(t)}_e\}$ and $X^o_{1,t} \triangleq \{X(s): s \in \mathbb{T}^{(t)}_o\}$, respectively. In order to use a coupling argument, define new stochastic process $\tilde{X}(\tau)$, $\tau \in \mathbb{R}$, such that for a fixed $\mathcal{R}_{a_{i-1}+1,a_i}$, $\{\tilde{X}(\tau): \tau \in \mathcal{R}_{a_{i-1}+1,a_i}\}$ and $\{X(\tau): \tau \in \mathcal{R}_{a_{i-1}+1,a_i}\}$ have the same distribution, $i=1,2,\ldots,2m$. Now, consider $\tilde{X}^h_{1,t} \triangleq \{\tilde{X}(s): s \in  \mathbb{T}^{(t)}_h\}$, $h \in \{e,o\}$, such that the requests in the even (and odd) blocks of $\tilde X_{1,t}$ are independent. However, within each block, the random variables can be arbitrarily correlated. We can always construct such a stochastic process, and the pair $({X}(s), \tilde{X}(s))$ is called a \emph{coupling} (see Fig. \ref{fig:figslot_timediv}). We define $\tilde X_{1,t}^e$ and $\tilde X_{1,t}^o$ similarly to $X_{1,t}^e$ and $X_{1,t}^o$, respectively.

The following theorem provides a bound on the performance guarantees in terms of the $\beta-$mixing coefficient.
\begin{thrm} \label{thm:main_res1}
For the given model, and the popularity estimate in \eqref{eq:estimation_popularity}, with a probability of at least $1-\delta$, the following holds
\begin{eqnarray}
\nonumber \hat{\mathcal{T}}{^\ast}(t + T)  \hspace{-0.3cm}&-&\hspace{-0.3cm} \mathcal{T}{^\ast}(t + T) < \min\{\mathbb{E}[\mathcal{{A}}_T(\tilde{{X}}^e_{1,t})], \mathbb{E}[\mathcal{{A}}_T(\tilde{{X}}^o_{1,t})]\} + \frac{N \alpha_{\texttt{max}} B a_{\texttt{max}}}{ \alpha_{\texttt{min}} R_0 a_{\texttt{min}}} \sqrt{\frac{\log \left(\frac{2}{\delta^{'}}\right)}{2 m}},
\label{eq:thm3eq}
\end{eqnarray}
where $\delta^{'} \triangleq \delta/2 - \exp\left\{-\lambda_u \pi R^2\right\} - \sum_{i=2}^{2m-1} \beta(a_i)  - e^{-\lambda_u \pi R^2} \sum_{j=1}^\infty \sum_{i=1}^{2m} (1-\zeta_{a_i,j}) \frac{(\lambda_u \pi R^2)^j}{j!} > 0$. Further, 
\begin{eqnarray}
\mathcal{A}_T(\tilde X^{(h)}_{1,t}) \triangleq \sup_{\Pi \in \mathcal{P}_\pi} \left\vert \sum_{i=1}^N g(\pi_i) \left(\hat{p}^h_{i,t}   - p_{i,t+T}\right)\right \vert,
\end{eqnarray}
where $\hat{p}^h_{i,t} \triangleq \frac{1}{\abs{\mathbb{T}^{(t)}_h}} \sum_{s \in \mathbb{T}^{(t)}_h} \mathds{1}\{\tilde X(s) = i\}$, $h \in \{e,o\}$.
\end{thrm}
\begin{proof}
See \appref{app:main_res1_proof}.
\end{proof}

Note that $\delta^{'} > 0$ implies a bound on $\delta$. Next, we bound $\min \{\mathbb{E}[\mathcal{{A}}_T(\tilde{{X}}^e_{1,t})], \mathbb{E}[\mathcal{{A}}_T(\tilde{{X}}^o_{1,t})]\}$ to get the desired result. The bound that we derive depends on the Rademacher complexity and the nonstationarity of the stochastic process. We begin with the following definition.
\defn{
\textbf{(Rademacher complexity)} The Rademacher complexity of $\mathcal{P}_\pi$ is defined as \cite[Chapter 3]{Mohri2012}
\begin{eqnarray*}
\mathcal{R}^{(t)}_h \triangleq  \mathbb{E}_{\tilde X, \bm{\sigma}} \frac{1}{\abs{\mathbb{T}^{(t)}_h}}\sup_{\Pi \in \mathcal{P}_\pi}\sum_{i=1}^N \!\! g(\pi_i)  \vert  \sum_{s \in \mathbb{T}^{(t)}_h} \!\!\sigma_{i,s} \mathds{1}\{\tilde X(s) = i\}
\vert,
\end{eqnarray*}
where the Rademacher random variables $\sigma_{i,s} \in \{-1,1\}$, $i=1,2,\ldots, N$ for $s \in \mathbb{T}^{(t)}_h$ are {\iid} with probability $1/2$, $\bm{\sigma} \triangleq \{\sigma_{i,s} \in \{-1,1\}: i=1,2,\ldots, N,  s \in \mathbb{T}^{(t)}_h\}$, and $h \in \{e,o\}$. \label{def:rademachercomplexity}
}

Next, we provide one of the main results of this paper. 
\begin{thrm} \label{thm:main_res2}
For the given model and the popularity estimate in \eqref{eq:estimation_popularity}, with a probability of at least $1-\delta$, the following holds:
\begin{eqnarray} \label{eq:pacbound_2}
\nonumber \hat{\mathcal{T}}{^\ast}(t + T)  <  \mathcal{T}{^\ast}(t + T) + 2 \max\{\mathcal{R}^{(t)}_e,\mathcal{R}^{(t)}_o\} + \max\{\Delta^{(e)}_{t,T}, \Delta^{(o)}_{t,T}\} + \frac{N \alpha_{\texttt{max}} B a_{\texttt{max}}}{R_0a_{\texttt{min}} \alpha_{\texttt{min}}} \sqrt{\frac{a_{\texttt{max}}\log \left(\frac{2}{\delta^{'}}\right)}{ t}},
\end{eqnarray}
where $\mathcal{R}^{(t)}_h$ is the Rademacher complexity, $a_{\texttt{max}} \triangleq \max_{1\leq i \leq 2m} a_i$, $\Delta^{(h)}_{t,T} \triangleq  \sup_{\Pi \in \mathcal{P}_\pi} \sum_{i=1}^N g(\pi_i)  d_i^{(h)}(t,T)$,
$d^{(h)}_i(t,T) \triangleq \frac{1}{\abs{\mathbb{T}^{(t)}_h}} \sum_{s \in \mathbb{T}^{(t)}_h} \left \vert  p_{i,s} - p_{i,t+T} \right \vert$, $h \in \{e,o\}$, and $\delta^{'} > 0$ is as defined in \thrmref{thm:main_res1} with $m = \lceil \frac{t}{a_\texttt{max}} \rceil$.
\end{thrm}
\begin{proof}
See \appref{app:mainres2proof}.
\end{proof}

\textbf{Remarks:}
\begin{enumerate}[(1)]
\item  The error $\epsilon$ increases linearly with $N$. To compensate for larger values of $N$, the waiting time $t$ should be of the order of $N^2$; a similar observation was also made in \cite{Bharath2016}. As $\lambda_u$ increases, a lower value of $\delta$ can be achieved. In general, as $\lambda_u \rightarrow \infty$, $\delta = 0$ cannot be achieved due to the dependence of the stochastic process across time, {\ie}, $\beta(a) > 0$, $a > 0$.

\item The error $\epsilon$ decreases as $t$ increases. When the requests are {\iid}, $a_{\texttt{max}} = 1$, and hence, $\epsilon$ is small. Thus, when the requests are correlated we incur a penalty of $a_{\texttt{max}}$, since the error decreases as $\sqrt{1/(t/a_{\texttt{max}})}$ compared to $\sqrt{1/t}$ for {\iid} requests. The error can be reduced by choosing $a_{\texttt{max}} =1$, {\ie}, $a_i = 1$, $i=1, \dots, 2m$. Since $\beta(x)$ is a monotonically decreasing function of $x$, the probability of achieving a lower error is very small, indicating a tradeoff between the error and the probability with which the bound in \eqref{eq:pacbound_2} holds. Also, lower values of $\delta^{'}$ result in higher error. This requires the value of $m$ to be small. However, $m$ scales as $t/a_\texttt{max}$, which indicates that if $a_\texttt{max} = \mathcal{O}(\sqrt t)$, then the last term in the error goes down as $1/t^{1/4}$ instead of $1/\sqrt t$. On the other hand, for larger values of $m$, the value of $\delta^{'}$ is small provided the $\beta$-mixing coefficient reduces at a smaller rate compared to $1/\sqrt{t}$; this indicates that one should have sufficiently fast decaying $\beta$-mixing for better performance. The last term in the expression for $\delta^{'}$ depends on $\zeta_{a_i,j}$, whose effect is studied by looking at specific examples, such as the Bernoulli and Poisson models for user requests, as detailed in the next section.

\item The error $\epsilon$ increases with $\frac{\alpha_{\texttt{max}}} {\alpha_{\texttt{min}}}$. The higher this ratio, the larger the variation in the number of requests. On the other hand, the lower this ratio, the smaller the error; which indicates a greater number of requests. The non-stationarity of the process is captured through $\Delta_{t,T}^{(h)}$, $h \in \{e,o\}$. For a stationary process $\Delta_{t,T}^{(h)} = 0$, $h\in \{e,o\}$.

\item When the user requests are {\iid}, the error does not vanish as $t \rightarrow \infty$, because the Rademacher complexity will not go to zero as $t \rightarrow \infty$. This indicates the difficulty in estimating the offloading loss, or equivalently the popularity profile, for a given caching policy.

\item The only term that depends on $T$ is $\max\{\Delta^{(e)}_{t,T}, \Delta^{(o)}_{t,T}\}$. The frequency with which the cache update should be done depends on $\Delta^{(h)}_{t,T}$, $h \in \{e,o\}$. For instance, if $\Delta^{(h)}_{t,T}$, $h \in \{e,o\}$, is high, then the updates should be more frequent.

\item The error is directly proportional to the number of bits per file, and inversely proportional to the rate at which the file is transmitted from the SBS to the users. 
\end{enumerate}

\section{Bernoulli and Poisson Requests}\label{sec:bernoulli_poisson}
In this section, we consider Bernoulli and Poisson request models, and analyze the implications on the results derived so far. 

\subsection{Bernoulli request model}
Let $X^{k}_u \in \{0,1\}$, $u \in \Phi_u$, denote the request made by user $u$ for a cached file, in the $k^{\text{th}}$ slot. In the Bernoulli model, it is assumed that $X^{k}_u \in \{0,1\}$ is i.i.d. across users and slots. Further, a user makes a request with probability $p$ in each time slot, independent of the file it requests, {\ie}, $\Pr\{X^{k}_u = 1\} = p$. The slot width $\Delta > 0$ is chosen such that at most one file is requested. Conditioned on the event that a set of  requests are made from several users, the files requested follow a non-stationary dependent random process. This simplified assumption makes the analysis of the offloading loss guarantees tractable. To provide theoretical guarantees for this model, from the general result in \thrmref{thm:main_res2}, it suffices to prove an upper bound on the probability of the event $\left\{r_{a_i} < \alpha_{\texttt{min}} n a_i\right\}  \bigcup \left\{r_{a_i}> \alpha_{\texttt{max}} n a_i\right\}$ in the $i$th block of size $a_i$, conditioned on the presence of $n$ users, {\ie}, 
\begin{eqnarray} \label{eq:firstterm_bern_LD}
\Pr\left\{r_{a_i} < \alpha_{\texttt{min}} n a_i  \bigcup r_{a_i} > \alpha_{\texttt{max}} n a_i \left \vert \right. n_u = n\right\} &\leq& \Pr\left\{r_{a_i} < \alpha_{\texttt{min}} n a_i    \left \vert \right. n_u = n\right\} \nonumber \\&&+ \Pr\left\{r_{a_i} > \alpha_{\texttt{max}} n a_i \left \vert \right. n_u = n\right\},
\end{eqnarray}
where $r_{a_i}$ is the total number of requests in $a_i$ slots, which is the sum of $a_i n$ independent Bernoulli random variables, leading to $\mathbb{E}[r_{a_i} \left \vert \right. n_u = n] = a_i np$. Towards this end, we use the following result:
\begin{thrm}
Let $X_1, X_2,\ldots,X_k$ be independent Bernoulli random variable with 
\begin{equation}
\Pr\{X_i = 1\} = p~~~~\Pr\{X_i=0\} = 1-p.
\end{equation}
Then, for $X \triangleq \sum_{i=1}^n X_i$ and $\lambda > 0$, we have
\begin{equation}
\Pr\{X \leq \mathbb{E}[X] - \lambda\} \leq \exp\{-\lambda^2/2np\},
\end{equation}
and 
\begin{equation}
\Pr\{X \geq \mathbb{E}[X] + \lambda\} \leq \exp\left\{-
\frac{\lambda^2}{2(np + \lambda/3)}\right\}.
\end{equation}
\label{thm:bernoulli1}
\end{thrm}
Using \thrmref{thm:bernoulli1} conditioned on the event $\{n_u = n\}$, we have the following theorem.
\begin{thrm} \label{thm:bern_largedev_bound}
For the Bernoulli model with $0 < p < \alpha_{\texttt{min}} <  \alpha_{\texttt{max}}$, we have 
\begin{equation}
\Pr\left\{r_{a_i} < \alpha_{\texttt{min}} n a_i  \bigcup r_{a_i} > \alpha_{\texttt{max}} n a_i \left \vert \right. n_u = n\right\} \leq 2 \exp\left\{- \frac{\psi_p a_{\texttt{min}} n}{2p} \right\},
\end{equation}
for $i=1,2,\ldots,2m$, and $n \geq 1$. In the above, $\psi_p \triangleq \min\left\{ \frac{a_{\texttt{min}}(p-\alpha_{\texttt{max}})^2}{1 + \frac{a_\texttt{max} (\alpha_{\texttt{min}} - p)}{3}}, (p-\alpha_{\texttt{min}})^2 \right\}$. 
\label{thm:bernoulli2}
\end{thrm}
\emph{Proof:} From \eqref{eq:firstterm_bern_LD}, it suffices to bound the following two terms $\Pr\left\{r_{a_i} < \alpha_{\texttt{min}} n a_i \left \vert \right. n_u = n\right\}$ and $\Pr\left\{r_{a_i} > \alpha_{\texttt{max}} n a_i \left \vert \right. n_u = n\right\}$. We start by upper bounding the first term in \eqref{eq:firstterm_bern_LD}. Using $\mathbb{E}[r_i \left \vert \right. n_u = n] = np a_i$ and choosing $\lambda \triangleq na_i(\alpha_{\texttt{min}} - p)$ in \thrmref{thm:bernoulli1} results in 
\begin{eqnarray}
\Pr\left\{r_{a_i} < \alpha_{\texttt{min}} n a_i    \left \vert \right. n_u = n\right\}  &\leq& \exp\left\{-\frac{(p-\alpha_{\texttt{min}})^2 a_i n}{2p}\right\}\nonumber\\ 
&\leq& \exp\left\{-\frac{(p-\alpha_{\texttt{min}})^2 a_{\texttt{min}} n}{2p}\right\},
\label{eq:bernoulli1}
\end{eqnarray}
for all $0 < p < \alpha_{\texttt{min}}$, and $i=1,2,\ldots,2m$.
Similarly, the second term in \eqref{eq:firstterm_bern_LD} can be bounded as
\begin{eqnarray}
\Pr\left\{r_{a_i} > \alpha_{\texttt{max}} n a_i    \left \vert \right. n_u = n\right\}  &\leq& \exp\left\{-\frac{(p-\alpha_{\texttt{max}})^2 a_i^2 n}{2(p + a_i(\alpha_{\texttt{max}} - p)/3)}\right\} \nonumber \\ 
&\leq& \exp\left\{-\frac{(p-\alpha_{\texttt{max}})^2 a_{\texttt{min}}^2 n}{2(p + a_{\texttt{max}}(\alpha_{\texttt{max}} - p)/3)}\right\} \nonumber \\
&\leq& \exp\left\{-\frac{(p-\alpha_{\texttt{max}})^2 a_{\texttt{min}}^2 n}{2p (1 + a_{\texttt{max}}(\alpha_{\texttt{max}} - p)/3p)}\right\} ,
\label{eq:bernoulli2}
\end{eqnarray}
for all $p < \alpha_{\texttt{max}}$ and any $i=1,2,\ldots,2m$. Combining \eqref{eq:bernoulli1} and \eqref{eq:bernoulli2} gives the desired result. This completes the proof of \thrmref{thm:bernoulli2}. $\blacksquare$

By using \thrmref{thm:bernoulli2}, we have $\Pr\left\{\alpha_{\texttt{min}} n a_i  < r_{a_i} < \alpha_{\texttt{max}} n a_i \left \vert \right. n_u = n\right\} \geq 1 - 2 \exp\left\{- \frac{\psi_p a_{\texttt{min}} n}{2p} \right\} \triangleq  \zeta_{a,n}$. Using this in the expression for $\delta^{'}$ in \thrmref{thm:main_res2}, and after some algebraic manipulation, we obtain the following result. 

\begin{thrm}
For the Bernoulli request model with $0 < p < \alpha_{\texttt{min}} < \alpha_{\texttt{max}}$, and the popularity estimate in \eqref{eq:estimation_popularity}, the following holds with a probability of at least $1-\delta$
\begin{eqnarray} \label{eq:pacbound_2bern}
\nonumber \hat{\mathcal{T}}{^\ast}(t + T)  \leq \mathcal{T}{^\ast}(t + T) + 2 \max\{\mathcal{R}^{(t)}_e,\mathcal{R}^{(t)}_o\} + \max\{\Delta^{(e)}_{t,T}, \Delta^{(o)}_{t,T}\} + \frac{N B a_{\texttt{max}}\alpha_{\texttt{max}}}{ a_{\texttt{min}} R_0 \alpha_{\texttt{min}}} \sqrt{\frac{a_{\texttt{max}} \log \left(\frac{2}{\delta^{'}}\right)}{ t}},
\end{eqnarray}
where $\mathcal{R}^{(t)}_h$ is the Rademacher complexity, and $$\Delta^{(h)}_{t,T} \triangleq  \sup_{\Pi \in \mathcal{P}} \sum_{i=1}^N g(\pi_i)  d_i^{(h)}(t,T),$$
$d^{(h)}_i(t,T) \triangleq \frac{1}{\abs{\mathbb{T}^{(t)}_h}} \sum_{s \in \mathbb{T}^{(t)}_h} \left \vert  p_{i,s} - p_{i,t+T} \right \vert$, $h \in \{e,o\}$. Further, $$\delta^{'} = \frac{\delta}{2} - \left(\exp\left\{-\lambda_u \pi R^2\right\} + \sum_{i=2}^{2m-1} \beta(a_i) + 4m \left[e^{-\lambda_u \pi R^2} (e^{-\lambda_u \pi R^2e^{-\phi_p}} - 1) \right]\right)> 0,$$ where $\phi_p \triangleq \frac{a_\texttt{min} \psi_p}{2p}$, and $\psi_p$ is as defined in Theorem \ref{thm:bern_largedev_bound}.
\end{thrm}

From the above theorem, the following observations can be made. First, assuming that $a_\texttt{min}$ and $a_\texttt{max}$ grow as $\mathcal{O}(\sqrt t)$ (which implies that $m = \mathcal{O}(\sqrt t)$), the last term in the error in \eqref{eq:pacbound_2bern} goes to zero as $1/t^{1/4}$, while the other terms are not effected by this choice. For $m = \mathcal{O}(\sqrt t)$, the second term in the expression for $\delta^{'}$ tends to zero as $t \rightarrow \infty$, provided that $\beta(\sqrt t) \rightarrow 0$ as $t \rightarrow \infty$. This demands a faster decay rate of $\beta$-mixing. The last term in the expression for $\delta^{'}$ tends to $-\infty$ as $t \rightarrow \infty$, resulting in a larger value of $\delta^{'}$, and hence, reducing the error. As a result of this, asymptotically in $t$, any value of $\delta > 0$ is a valid choice. Thus, by choosing $\delta$ sufficiently close to $0$, a high probability result on the performance can be obtained.   

\subsection{Poisson request model} \label{sec:poiss}
We assume that the requests follow a Poisson model as defined below.\\
\texttt{\textbf{Assumption~$2$}:} The number of requests across users in any interval follows an independent homogenous Poisson process with arrival rate $\lambda_r$. Conditioned on the number of requests, the requested files follow a non-stationary, possibly dependent stochastic process.

As in the previous subsection, we first provide a bound on $\zeta_{a_i,n}$ for each $i$. 
\begin{thrm}
For the Poisson request model, with $\alpha_{\texttt{min}} = \frac{\Delta \lambda_r}{e^2}$ and $\alpha_{\texttt{max}} = {\Delta \lambda_r e}$, the following bound holds
\begin{eqnarray}
\Pr\left\{r_{a_i} < \alpha_{\texttt{min}} n a_i  \bigcup r_{a_i} > \alpha_{\texttt{max}} n a_i \left \vert \right. n_u = n\right\} \leq 2\exp\{-n a_\texttt{min} \lambda_r \Delta\}.
\label{eq:poisson1}
\end{eqnarray}
\end{thrm}
\emph{Proof:} First, consider the following with $\tau \triangleq \alpha_{\texttt{min}} n a_i$
\begin{eqnarray}
\Pr\left\{r_{a_i} < \tau \left \vert \right. n_u = n\right\} &=& \Pr\left\{e^{-s r_{a_i}} > e^{-\tau s}\left \vert \right. n_u = n\right\}
\leq \inf_{s > 0} e^{\tau s} \mathbb{E} [e^{-r_{a_i} s}\left \vert \right. n_u = n]\nonumber\\
&\leq& \exp\left\{ - n a_i \left[\Delta \lambda_r - \alpha_{\texttt{min}} \left(1 - \log\left(\frac{\Delta \lambda_r}{\alpha_{\texttt{min}}}\right) \right) \right] \right\},
\end{eqnarray}
where the last inequality follows by using Chernoff bound along with the fact that $\mathbb{E} [r_{a_i}] = \lambda_r \Delta n a_i$.
Substituting for $\tau$, using $\alpha_{\texttt{min}} = \frac{\Delta \lambda_r}{e^2}$, and the fact that $a_i \geq a_\texttt{min}$ for all $i$, we get
\begin{equation}
\Pr\left\{r_{a_i} < \tau \left \vert \right. n_u = n\right\} \leq \exp\left\{-n a_\texttt{min} \lambda_r \Delta \left(1+ \frac{1}{e^2}\right) \right\}. \label{eq:poiss_bound1}
\end{equation}
Now, consider the following term:
\begin{eqnarray}
\Pr\left\{r_{a_i} > \alpha_{\texttt{max}} n a_i \left \vert \right. n_u = n\right\} &\leq& \exp\left\{-na_i \lambda_r \Delta \left( 1 - \frac{\alpha_{\texttt{max}}}{\lambda_r \Delta} + \frac{\alpha_{\texttt{max}}}{\lambda_r \Delta} \log \left(\frac{\alpha_{\texttt{max}}}{\lambda_r \Delta}\right) \right) \right\} \nonumber \\
&\leq& \exp\{- n a_\texttt{min} \lambda_r \Delta\}, \label{eq:poiss_bound2}
\end{eqnarray}
where the inequality follows from the Chernoff bound, and the last inequality follows by choosing $\alpha_{\texttt{max}} = {e \Delta \lambda_r} > \alpha_{\texttt{min}} = \Delta \lambda_r / e^2$, and $a_i \geq a_{\texttt{min}}$. From \eqref{eq:poiss_bound1} and \eqref{eq:poiss_bound2}, we get the bound in \eqref{eq:poisson1}. $\blacksquare$

\begin{thrm}
For the Poisson request model with the popularity estimate in \eqref{eq:estimation_popularity}, with a probability of at least $1-\delta$, the following holds
\begin{eqnarray} \label{eq:pacbound_2}
\nonumber \hat{\mathcal{T}}{^\ast}(t + T)  \leq \mathcal{T}{^\ast}(t + T) + 2 \max\{\mathcal{R}^{(t)}_e,\mathcal{R}^{(t)}_o\} + \max\{\Delta^{(e)}_{t,T}, \Delta^{(o)}_{t,T}\} + \frac{N B a_{\texttt{max}}e}{ a_{\texttt{min}} R_0} \sqrt{\frac{a_{\texttt{max}} \log \left(\frac{2}{\delta^{'}}\right)}{ t}},
\end{eqnarray}
where $\mathcal{R}^{(t)}_h$ is the Rademacher complexity, $$\Delta^{(h)}_{t,T} \triangleq  \mathbb{E}\left[\sup_{\Pi \in \mathcal{P}} \sum_{i=1}^N g(\pi_i)  d_i^{(h)}(t,T)\right],$$
where $d^{(h)}_i(t,T) \triangleq \frac{1}{\abs{\mathbb{T}^{(t)}_h}} \sum_{s \in \mathbb{T}^{(t)}_h} \left \vert  p_{i,s} - p_{i,t+T} \right \vert$, $h \in \{e,o\}$. Further, 
$$\delta^{'} = \frac{\delta}{2} - \left(\exp\left\{-\lambda_u \pi R^2\right\} + \sum_{i=2}^{2m-1} \beta(a_i) + 4m \left[e^{-\lambda_u \pi R^2} (e^{-\lambda_u \pi R^2e^{-a_\texttt{min} \lambda_r \Delta}} - 1) \right]\right)> 0.$$

\end{thrm}

As in the Bernoulli case, a better performance can be achieved by choosing $m = \mathcal{O}(\sqrt t)$ and $a_i = \mathcal{O}(\sqrt t)$ for all $i$. It can also be seen that as $\lambda_r$ (and $\Delta$) increases, a smaller value of $\delta$ is possible leading to a better performance. However, unlike the Bernoulli model, the bound is independent of $\alpha_{\texttt{min}}$ and $\alpha_{\texttt{max}}$. The results presented for the models considered here lead to a simple yet effective algorithm for updating the cache when the popularity profile is varying across time. Next, we provide the details of this algorithm along with numerical simulations.

\section{Cache Update Algorithm and Numerical Results}\label{sec:numerical_results}
In this section, we present a cache update algorithm following \thrmref{thm:main_res2}, and the corresponding simulation results. \thrmref{thm:main_res2} suggests that the sBSs should update their caches at the time instants at which the error becomes large. The only relevant term is $\max\{\Delta^{(e)}_{t,T},\Delta^{(o)}_{t,T}\} \leq \Delta_{t,T} \triangleq \frac{1}{\abs{\mathbb{T}^{(t)}_e \bigcup \mathbb{T}^{(t)}_o}}\sup_{\Pi \in \mathcal{P}_\pi} \sum_{i=1}^N    \sum_{s \in \mathbb{T}^{(t)}_o \bigcup \mathbb{T}^{(t)}_e} g(\pi_i)  \left \vert  p_{i,s} - p_{i,t+T} \right \vert$. The following cache update mechanism is employed:
\begin{enumerate}
\item \small{Initialize $t=0$ and $T=0$. Update the caches randomly.
\item If $\hat\Delta_{t,T} > \texttt{threshold}$, then update the caches using the caching probability obtained by solving $\hat{\Pi}^{*}_{t+T} = \arg \min_{\Pi^{(t+T)} \in \mathcal{P}_{\pi}} ~~ {\mathcal{T}}(\Pi^{(t+T)},\hat{\mathcal{P}}^{(t+T - 1)})$, where $\hat{\mathcal{P}}^{(t + T-1)}$ is the estimate obtained using \eqref{eq:estimation_popularity}, and set $T = t$. Here, $\hat\Delta_{t,T}$ denotes an estimate of $\Delta_{t,T}$, and $\texttt{threshold} > 0$ determines the error achieved. 
\item Set $t \leftarrow t+1$ and go to step $2$.}
\end{enumerate} 

\begin{figure}[h!]
\begin{center}
{\includegraphics[height=3.5in,width=4.5in]{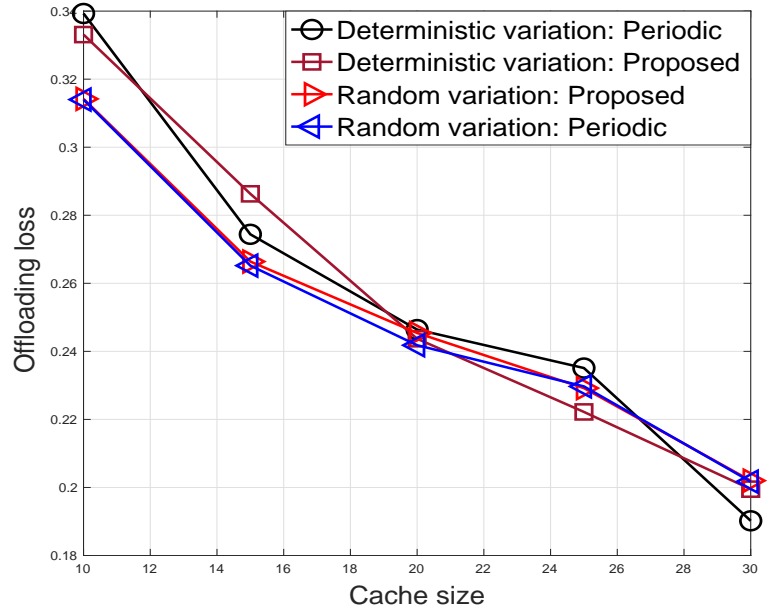}}
\caption{Offloading loss as a function of the cache size.}
\label{fig:offload_vs_cachesize}
\end{center}
\end{figure}
\begin{figure}[h!]
\begin{center}
{\includegraphics[height=3.5in,width=4.5in]{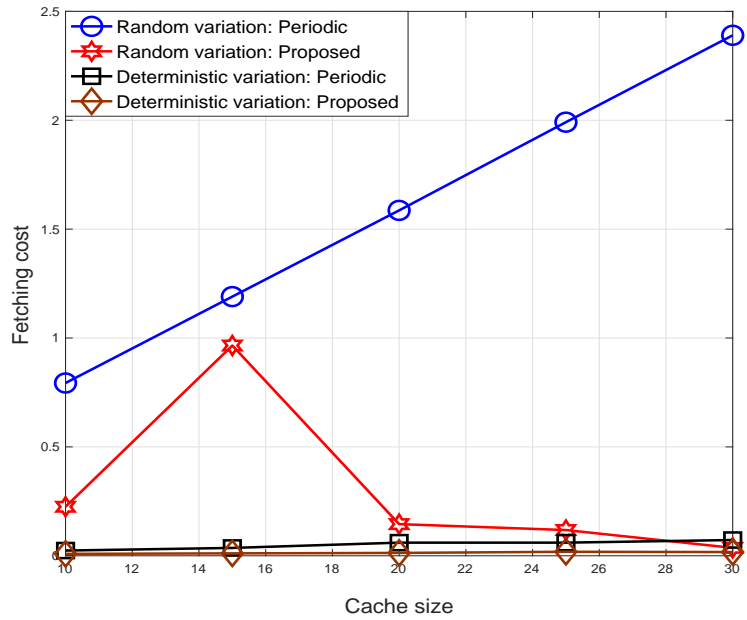}}
\caption{Fetching cost versus cache size for two different scenarios of arrival process.}
\label{fig:fetching_cost_vs_cachesize}
\end{center}
\end{figure}

We define the fetching cost as the average number of files downloaded at each cache update. The simulation setup consists of sBSs and users distributed according to PPPs with densities $\lambda_B = 0.00001$ and $\lambda_u = 0.0001$, respectively. The number of files is $N=100$, and the coverage of the BS and sBSs are $1000$ m and $500$ m, respectively. We let $\gamma= 500$. The deterministic arrival rate corresponds to a deterministic variation in the distribution of the popularity profile once every $150$ slots; while the random change corresponds to a random change in the popularity profile which occurs once every $100$ slots on average. In the deterministic variation scenario, a random set of $3$ pairs of files are chosen, and are permuted in a uniformly random fashion.

In the random variation scenario, two pairs of indices are randomly and uniformly permuted at random times. The requests follow a Poisson arrival model with rates $\lambda_r = 0.09$ and $0.01$ for the scenarios corresponding to random and deterministic changes, respectively. Requests for the files are generated using a Zipf distribution with parameter $\theta = 0.8$. Thus, the arrival is non-stationary but independent across time. This non-stationarity results in oscillations in the curves. The requests from a typical user at the origin are used to evaluate the offloading loss. \figref{fig:offload_vs_cachesize} shows the offloading loss with $B = R_0$ as a function of the cache size for the two scenarios mentioned above. The periodic updates are carried out every $5$ time slots. It is clear from the figure that, for the random variation scenario, the performance of the proposed scheme and the periodic scheme are almost the same. However, we observe in \figref{fig:fetching_cost_vs_cachesize} that the fetching cost of the proposed scheme is lower, as the periodic update scheme requires far too many updates. This confirms that by appropriately choosing the \texttt{threshold} values, the proposed scheme outperforms the periodic cache update scheme for specific scenarios. The variation in the fetching cost for the proposed (deterministic) scheme is an artifact of choosing the \texttt{threshold}. For the deterministic variation case, it can be seen in \figref{fig:fetching_cost_vs_cachesize} that for certain cache sizes ($10, 20$ and $25$), the offloading loss of the proposed scheme outperforms periodic caching, while it performs poorly for other cache sizes. However, the fetching cost is lower than that of the periodic update scheme for all the cache sizes . This shows that in order to achieve a smaller offloading loss, it is better to update more frequently; while in other scenarios (cache size = $15$), it is possible to achieve both a lower offloading loss and a lower fetching cost. A smaller offloading loss can be achieved by lowering the \texttt{threshold} value at the expense of the fetching cost. The gain of the proposed scheme depends on how frequently the popularity profile changes. For example, when the popularity profile changes slowly, the gain is small; but the frequency of updates will also be less in the proposed scheme.

\section{Concluding remarks}\label{sec:conclusion}
A learning-theoretic analysis of content caching in heterogenous networks with non-stationary, statistically dependent and unknown popularity profiles has been considered. A PAC result on the offloading loss is presented in \thrmref{thm:main_res2}, based on the following caching algorithm: At every slot $t$, the BS computes an estimate of the Rademacher complexity and the discrepancy based on the available requests. The optimal caching policy is employed at the BS based on these estimates, and the cache content items at the sBSs are updated only if the discrepancy in the popularity profile is larger than a pre-specified threshold (to be determined based on the error tolerance). A detailed analysis of this algorithm is relegated to future work. We also presented the performance analyses for the Bernoulli and Poisson request models.

\appendices
\vspace{-0.2in}

\section{Proof of \thrmref{thm:machine_learning}}\label{app:proof_ml}
First, we let $\hat{\mathcal{T}}^* \triangleq \mathcal{T}(\hat{\Pi}_t^*,{\mathcal{P}}^{(t+T)})$, $\hat{\mathcal{T}} \triangleq \mathcal{T}({\Pi},\hat{\mathcal{P}}^{(t)})$. Now consider the term $\hat{\mathcal{T}}^* - \inf_{\Pi} \mathcal{T}(\Pi,{\mathcal{P}}^{(t+T)})$. We can write
    \begin{eqnarray}
    \hat{\mathcal{T}}^* - \inf_{\Pi} \mathcal{T}(\Pi,{\mathcal{P}}^{(t+T)}) &=& \hat{\mathcal{T}}^* - \hat{\mathcal{T}} + \hat{\mathcal{T}} -\inf_{\Pi} \mathcal{T}(\Pi,{\mathcal{P}}^{(t+T)}) \nonumber\\
    &\leq& \hat{\mathcal{T}}^* - \hat{\mathcal{T}} + \sup_\Pi\mathcal{T}({\Pi},\hat{\mathcal{P}}^{(t)}) -\inf_{\Pi} \mathcal{T}(\Pi,{\mathcal{P}}^{(t+T)}) \nonumber\\
    &\leq& \hat{\mathcal{T}}^* - \hat{\mathcal{T}} + \sup_\Pi (\mathcal{T}( \Pi,\hat{\mathcal{P}}^{(t)}) - \mathcal{T}(\Pi,{\mathcal{P}}^{(t+T)})) \nonumber\\
    &\leq& \hat{\mathcal{T}}^* - \hat{\mathcal{T}} + \sup_\Pi \abs{\mathcal{T}( \Pi,\hat{\mathcal{P}}^{(t)}) - \mathcal{T}(\Pi,{\mathcal{P}}^{(t+T)})} \nonumber\\
    &\leq& \mathcal{T}({\hat{\Pi}}_t^*, \mathcal{P}^{(t+T)}) - \inf_{\Pi} \mathcal{T}( \Pi,\hat{\mathcal{P}}^{(t)}) + \sup_\Pi \abs{\mathcal{T}( \Pi,\hat{\mathcal{P}}^{(t)}) - \mathcal{T}(\Pi,{\mathcal{P}^{(t+T)}})} \nonumber\\
    &\leq& \sup_\Pi \mathcal{T}({{\Pi}}, \mathcal{P}^{(t+T)}) - \inf_{\Pi} \mathcal{T}( \Pi,\hat{\mathcal{P}}^{(t)}) + \sup_\Pi \abs{\mathcal{T}( \Pi,\hat{\mathcal{P}}^{(t)}) - \mathcal{T}(\Pi,{\mathcal{P}}^{(t+T)})} \nonumber\\
    &\leq&  \sup_\Pi \abs{\mathcal{T}({\Pi},\mathcal{P}^{(t+T)}) -  \mathcal{T}( \Pi,\hat{\mathcal{P}}^{(t)})} + \sup_\Pi \abs{\mathcal{T}( \Pi,\hat{\mathcal{P}}^{(t)}) - \mathcal{T}(\Pi,{\mathcal{P}}^{(t+T)})}. \nonumber\\
    &\leq& 2\sup_\Pi \abs{ \mathcal{T}({{\Pi}}, \mathcal{P}^{(t+T)}) - \mathcal{T}( \Pi,\hat{\mathcal{P}}^{(t)})},
    \end{eqnarray}
where all the inequalities above are self evident.

\section{Proof of \thrmref{thm:main_res1}} \label{app:main_res1_proof}
Consider the following:
\begin{eqnarray}
\nonumber \mathcal{A}_T(X_{1,t})
&\stackrel{(a)}{\leq}& \sup_{\Pi \in \mathcal{P}} \left\vert\frac{\abs{\mathbb{T}^{(t)}_e}}{r_t}\sum_{i=1}^N g(\pi_i) \left(\hat{p}^e_{i,t}   - p_{i,t+T}\right)\right \vert  + \sup_{\Pi \in \mathcal{P}} \left \vert \frac{{\abs{\mathbb{T}^{(t)}_o}}}{r_t}\sum_{i=1}^N g(\pi_i) \left(\hat{p}^o_{i,t}  - p_{i,t+T}\right) \right\vert \\
&\stackrel{(b)}{\leq}& \frac{\abs{\mathbb{T}^{(t)}_e}}{r_t} \mathcal{A}_T(X^e_{1,t})  + \frac{\abs{\mathbb{T}^{(t)}_o}}{r_t} \mathcal{A}_T(X^o_{1,t}), \label{eq:At}
\end{eqnarray}
where $\hat{p}^h_{i,t} \triangleq \frac{1}{\abs{\mathbb{T}^{(t)}_h}} \sum_{s \in \mathbb{T}^{(t)}_h} \mathds{1}\{X(s) = i\}$, $h \in \{e,o\}$, and
$\mathcal{A}_T(X^{(h)}_{1,t}) \triangleq \sup_{\Pi \in \mathcal{P}} \left\vert \sum_{i=1}^N g(\pi_i) \left(\hat{p}^h_{i,t}   - p_{i,t+T}\right)\right \vert$.
In \eqref{eq:At}, $(a)$ follows from algebraic manipulation and the triangle inequality, and $(b)$ follows from the convexity property. Using \eqref{eq:At}, and the union bound, we can write
\begin{eqnarray}
\nonumber \Pr \left\{ \mathcal{A}_{T}(X_{1,t})  > \epsilon  | n_u = j\right\} &\leq&  \Pr \left\{ \frac{\abs{\mathbb{T}^{(t)}_e}}{r_t} \mathcal{A}_T^e(X_{1,t}) + \frac{\abs{\mathbb{T}^{(t)}_o}}{r_t} \mathcal{A}_T^o(X_{1,t})   > \epsilon | n_u = j\right\} \\
\nonumber  &\stackrel{(a)}{\leq}&  \Pr \{ \mathcal{A}_T(X^e_{1,t})    > \epsilon  |  n_u = j\} +  \Pr \{ \mathcal{A}_T(X^o_{1,t})    > \epsilon  | n_u = j\},
\label{eq:evenodd_decouple}
\end{eqnarray}
where $(a)$ follows from the union bound. We now bound the term corresponding to the even samples. (The bound on the term corresponding to the odd samples can be obtained similarly, and is not shown here for sake of brevity). We begin with $\Pr \{ \mathcal{A}_T(X^e_{1,t}) > \epsilon  | n_u = j\} \!\! = \!\! \mathbb{E}[ \mathds{1}\{\mathcal{A}_T(X^e_{1,t}) > \epsilon\} {| n_u = j} ]$.
Since the indicator function is bounded, using \cite[Proposition 1]{Kuznetsov2014}, we have the following upper bound:
\begin{eqnarray}
\nonumber \mathbb{E}[ \mathds{1}\{\mathcal{A}_T(X^e_{1,t}) > \epsilon\} {| n_u = j} ] &\leq&  \mathbb{E}[ \mathds{1}\{\mathcal{{A}}_T(\tilde{{X}}^e_{1,t}) > \epsilon\} {| n_u = j}] + \sum_{i=2}^{m} \beta(a_{2i - 1}),
\nonumber \\
&=&  \Pr\{\mathcal{{A}}_T(\tilde{{X}}^e_{1,t}) > \epsilon{| n_u = j}  \} + \sum_{i=2}^{m} \beta(a_{2i - 1}),\label{eq:odd_coupling}
\end{eqnarray}
where $\tilde{{X}}^e_{1,t}$ is defined in \secref{sec:mainresult_1}. Since the conditioning is on $\{n_u = j\}$, the time slot difference between adjacent even/odd block is deterministic, and the $\beta$-mixing is not conditioned on the event. Similarly, it can be shown that
\begin{eqnarray}
 \mathbb{E}[ \mathds{1}\{\mathcal{A}_T(X^o_{1,t}) > \epsilon\} {| n_u = j}  ] \leq  \Pr\{ \mathcal{{A}}_T(\tilde{{X}}^o_{1,t}) > \epsilon{| n_u = j}  \} + \sum_{j=1}^{m-1} \beta(a_{2j}),
\label{eq:even_coupling}
\end{eqnarray}
where $\mathcal{{A}}_T(\tilde{{X}}^e_{1,t})$ (resp. $\mathcal{{A}}_T(\tilde{{X}}^o_{1,t})$) is obtained by replacing each block of data in $X^e_{1,t}$ (resp. $X^o_{1,t}$) by $\tilde{X}^e_{1,t}$ (resp. $\tilde{X}^o_{1,t}$) in the definition of $\mathcal{{A}}_T(X^e_{1,t})$ (resp. $\mathcal{{A}}_T(X^o_{1,t})$). Using \eqref{eq:even_coupling} in \eqref{eq:evenodd_decouple}, we get
\begin{eqnarray}
 \Pr \{ \mathcal{A}_{T}(X_{1,t}) > \epsilon  | n_u = j\}  \leq \!\!\!\! \sum_{h \in \{e,o\}}  \Pr\{ \mathcal{{A}}_T(\tilde{{X}}^h_{1,t}) > \epsilon{| n_u = j}  \}+
 \sum_{j=2}^{2m-1} \beta(a_{j}).
\label{eq:bound_evenodd}
\end{eqnarray}
Since each of the events involves sum of blocks of independent data, we employ McDiarmid's inequality to bound the probability in \eqref{eq:bound_evenodd}, as shown below.
\begin{thrm}
For any $\max \{\mathbb{E}[\mathcal{{A}}_T(\tilde{{X}}^e_{1,t})], \mathbb{E}[\mathcal{{A}}_T(\tilde{{X}}^o_{1,t})]\} < \epsilon$, and $m > 0$, the following bound holds for all $j \geq 1$:
\begin{eqnarray}
 \sum_{h \in \{e,o\}} \Pr\{\mathcal{{A}}_T(\tilde{{X}}^h_{1,t}) > \epsilon{| n_u = j}\}  \leq  2\exp\left\{-2 m g_{t,N} \right\} +   \sum_{i=1}^m\zeta_{a_i,j} \Pr\{n_u  = j\},
\label{eq:upperbound_mcdiarmid}
\end{eqnarray}
where $g_{t,N} \triangleq \frac{R_0^2 a_{\texttt{min}}^2 \min\{\epsilon_{e}^2, \epsilon_{o}^2\} \alpha_{\texttt{min}}^2}{a_{\texttt{max}}^2 B^2\alpha_{\texttt{max}}^2 N^2}$, $a_{\texttt{min}} \triangleq \min_{1\leq i \leq 2m} a_i$, $a_{\texttt{max}} \triangleq \max_{1\leq i \leq 2m} a_i$, and $\epsilon_{h} \triangleq \epsilon - \mathbb{E}[\mathcal{{A}}_T(\tilde{{X}}^h_{1,t})]$, $h \in \{e,o\}$.
\label{thm:mcdiarmid}
\end{thrm}
\begin{proof}
Consider the term corresponding to the even blocks
\begin{eqnarray} \label{eq:large_dev_app1}
\Pr\left\{\mathcal{{A}}_T(\tilde{{X}}^e_{1,t}) > \epsilon \left \vert \right. n_u = j\right\}   =  \Pr\left\{\mathcal{{A}}_T(\tilde{{X}}^e_{1,t}) - \mathbb{E}\left\{\mathcal{{A}}_T(\tilde{{X}}^e_{1,t})\right\}    > \epsilon_e \left \vert \right. n_u = j\right\},
\end{eqnarray}
where $\epsilon_e$ is as defined in the theorem.
To apply Mcdiarmid's inequality, we let $\tilde{{X}}_{1,t}^e$ and $\hat{{X}}_{1,t}^e$ be independent sequences of even blocks that differ only  in one block, say the $i$th block $a_i$. Let the distributions of $\tilde{{X}}_{1,t}^e$ and $\hat{{X}}_{1,t}^e$ be identical. Conditioned on $\{n_u = j\}$, let $s_{ik}$, $k=1,2,\ldots,a_i$ denote the number of requests in the $k$th slot of the $i$th block consisting of $a_i$ slots. Therefore, conditioned on $\{n_u = j\}$, we have
\begin{eqnarray}
\nonumber  \sup_{\Pi \in \mathcal{P}} \left\vert\tilde{g}_{t,T}(\tilde{{X}}_{1,t}^e)\right\vert - \sup_{\Pi \in \mathcal{P}} \left\vert\hat{g}_{t,T}(\hat{{X}}_{1,t}^e)\right\vert
&\stackrel{(a)}{\leq}& \sup_{\Pi \in \mathcal{P}} \biggl\vert\sum_{j=1}^N g(\pi_j) \biggl(\frac{1}{\abs{\mathbb{T}^{(t)}_e}} \sum_{s \in \mathbb{T}^{(t)}_e} \mathds{1}\{\tilde{X}(s) = j\} \biggr.\biggr.  - \biggl.\biggl.\mathds{1}\{\hat{X}(s) = j\}\biggr)\biggr\vert \nonumber \\
&\stackrel{(b)}{\leq}& \sup_{1\leq j \leq N} g(\pi_j)  \frac{N \sum_{k=1}^{a_i} s_{ik}}{\abs{\mathbb{T}^{(t)}_e}} \leq \frac{B N \sum_{k=1}^{a_i} s_{ik}}{R_0\abs{\mathbb{T}^{(t)}_e}}, \label{eq:bound_event}
\end{eqnarray}
where $(a)$ follows from the reverse triangle inequality, and $(b)$ follows from the fact that the two sequences $\tilde{{X}}_{1,t}^e$ and $\tilde{{X}}_{1,t}^o$ differ in the $i$th block, and the $i$th block can have at most $\sum_{k=1}^{a_i} s_{ik}$ requests. Further,
$$\hat{g}_{t,T}(\tilde{{X}}_{1,t}^e) \triangleq \sum_{i=1}^N g(\pi_i) \left(\frac{1}{\abs{\mathbb{T}^{(t)}_e}} \sum_{s \in \mathbb{T}^{(t)}_e}\mathds{1}\{\tilde{X}(s) = i\}   - p_{i,t+T}\right),$$ and $\hat{g}_{t,T}(\hat{{X}}_{1,t}^e) $ is defined in a similar fashion. Also, note that $\abs{\mathbb{T}^{(t)}_e} = \sum_{i=1}^{m} \sum_{k=1}^{a_i} s_{ik}$. Now, conditioned on the event that the number of requests in the $i$th block is bounded, i.e., $\mathcal{E}_{j} \triangleq \bigcap_{i=1}^m \left\{\alpha_{\texttt{min}} j a_i \leq r_i \leq \alpha_{\texttt{max}} j a_i \right\}$, we can write \eqref{eq:large_dev_app1} as 
\begin{eqnarray}
\Pr\left\{\mathcal{{A}}_T(\tilde{{X}}^e_{1,t}) - \mathbb{E}\left\{\mathcal{{A}}_T(\tilde{{X}}^e_{1,t})\right\}    > \epsilon_e \left \vert \right. n_u = j\right\} &\leq& \Pr\left\{\mathcal{{A}}_T(\tilde{{X}}^e_{1,t}) - \mathbb{E}\left\{\mathcal{{A}}_T(\tilde{{X}}^e_{1,t})\right\}    > \epsilon_e \left \vert \right. \mathcal{E}_{j}, n_u = j\right\} \nonumber \\ &\times& \Pr \{\mathcal{E}_{j} \left \vert \right. n_u=j\} + \Pr\{\mathcal{E}_{j}^c \left \vert \right. n_u = j\}, \nonumber\\
&\leq& \Pr\left\{\mathcal{{A}}_T(\tilde{{X}}^e_{1,t}) - \mathbb{E}\left\{\mathcal{{A}}_T(\tilde{{X}}^e_{1,t})\right\}    > \epsilon_e \left \vert \right.\mathcal{E}_{j}, n_u = j\right\} \nonumber \\
&& + \sum_{i=1}^m\zeta_{a_i,j},
\end{eqnarray} 
where the last inequality above follows from the union bound and \defref{def:rademachercomplexity}. Using \eqref{eq:bound_event}, and the fact that the event $\mathcal{E}_{j}$ occurs, we have
\begin{eqnarray}
\frac{B^2 N^2 \sum_{i=1}^m\left(\sum_{k=1}^{a_i} s_{ik}\right)^2}{R_0^2\abs{\mathbb{T}^{(t)}_e}^2} \leq \frac{B^2N^2 m\left(\alpha_{\texttt{max}} j a_{\texttt{max}}\right)^2}{R_0^2\left(\alpha_{\texttt{min}} j a_{\texttt{min}} m\right)^2} = \frac{B^2 N^2 \alpha_{\texttt{max}}^2 a_{\texttt{max}}^2}{R_0^2 \alpha_{\texttt{min}}^2 a_{\texttt{min}}^2 m},
\end{eqnarray}
where $a_{\texttt{min}} \triangleq \min_{1\leq i \leq 2m} a_i$ and $a_{\texttt{max}} \triangleq \max_{1 \leq i \leq 2m} a_i $. 
Using this boundedness property along with Mcdiarmid's inequality, we have
\begin{eqnarray}
\nonumber \Pr\left\{\mathcal{{A}}_T(\tilde{{X}}^e_{1,t}) - \mathbb{E}\left\{\mathcal{{A}}_T(\tilde{{X}}^e_{1,t})\right\} > \epsilon_e \left \vert \right.\mathcal{E}_{j}, n_u = j\right\}\leq  \exp\left\{-\frac{2a_{\texttt{min}}^2 R_0^2\alpha_{\texttt{min}}^2 m} {\epsilon_e^2 B^2 N^2 a_{\texttt{max}}^2 \alpha_{\texttt{max}}^2}\right\} + \sum_{i=1}^m\zeta_{a_i,j}.
\label{eq:even_mcdiarmid_ineq}
\end{eqnarray}
Similarly, $$\Pr\left\{\mathcal{{A}}_T(\tilde{{X}}^o_{1,t}) - \mathbb{E}\left\{\mathcal{{A}}_T(\tilde{{X}}^o_{1,t})\right\} > \epsilon_e \left \vert \right. \mathcal{E}_{j}, n_u = j\right\} \leq  \exp\left\{-\frac{2R_0^2 a_{\texttt{min}}^2\alpha_{\texttt{min}}^2 m} {\epsilon_o^2 B^2 N^2 a_{\texttt{max}}^2\alpha_{\texttt{max}}^2}\right\}+ \sum_{i=1}^m\zeta_{a_i,j}.$$ Combining these two, we get the desired result, which completes the proof of \thrmref{thm:mcdiarmid} and hence \thrmref{thm:main_res1}. 
\end{proof}

The bound in \eqref{eq:upperbound_mcdiarmid} is independent of $j$. From  \eqref{eq:upperbound_mcdiarmid}, \eqref{eq:bound_evenodd}, and using the result in \eqref{eq:first_bound}, we get
\begin{eqnarray}
\Pr\left\{\mathcal{A}_{T}(X_{1,t})  > \epsilon \right\} 
\leq \exp\left\{-\lambda_u \pi R^2\right\}  + \exp\left\{-\psi m \right\} + \sum_{i=2}^{2m-1} \beta(a_i) + e^{-\lambda_u}\sum_{j=1}^\infty\sum_{i=1}^m\zeta_{a_i,j} \frac{\lambda_u^j}{j!},
\end{eqnarray}
where $\psi \triangleq \frac{2  a_{\texttt{max}}^2 \min\{\epsilon_{e}^2, \epsilon_{o}^2\} R_0^2\alpha_{\texttt{min}}^2}{a_{\texttt{min}}^2 \alpha_{\texttt{max}}^2 N^2 B^2}$. We need $\Pr\left\{\mathcal{A}_{T}(X_{1,t})  > \epsilon \right\} < \delta/2$, which implies that
\begin{eqnarray}
\min\{\epsilon_{e}, \epsilon_{o}\} > \frac{N B a_{\texttt{max}}\alpha_{\texttt{max}}}{a_{\texttt{min}} R_0 \alpha_{\texttt{min}}} \sqrt{\frac{\log \left(\frac{2}{\delta^{'}}\right)}{2 m}},
\label{eq:thrm5eq1}
\end{eqnarray}
where 
\begin{equation}
\delta^{'} \triangleq \delta/2 - \exp\left\{-\lambda_u \pi R^2\right\} - \sum_{i=2}^{2m-1} \beta(a_i) - e^{-\lambda_u}\sum_{j=1}^\infty\sum_{i=1}^m\zeta_{a_i,j} \frac{\lambda_u^j}{j!} > 0. 
\end{equation}
But, $\epsilon_h = \epsilon - \mathbb{E}\left[\mathcal{{A}}_T(\tilde{{X}}^h_{1,t})\right]$, $h \in \{e,o\}$. Using this in \eqref{eq:thrm5eq1} results in the following constraint:
\begin{eqnarray}
\epsilon > \mathcal{E}_{t,T}  + \frac{N B a_{\texttt{max}}\alpha_{\texttt{max}}}{R_0 a_{\texttt{min}} \alpha_{\texttt{min}}} \sqrt{\frac{\log \left(\frac{2}{\delta^{'}}\right)}{2 m}},
\label{eq:epsilon1}
\end{eqnarray}
where $\mathcal{E}_{t,T} \triangleq \min\left\{\mathbb{E}\left[\mathcal{{A}}_T(\tilde{{X}}^e_{1,t})\right], \mathbb{E}\left[\mathcal{{A}}_T(\tilde{{X}}^o_{1,t})\right]\right\}$. With probability of at least $(1-\delta)$, $\mathcal{T}{^\ast}(t + T)  <  \mathcal{T}{^\ast}(t + T) < \epsilon$ implies the bound in the theorem after substituting for $\epsilon$ in \eqref{eq:epsilon1}.

\section{Proof of \thrmref{thm:main_res2}} \label{app:mainres2proof}
We only consider the term $\mathbb{E}[\mathcal{{A}}_T(\tilde{{X}}^e_{1,t})]$, since an upper bound on the other term follows similarly. As before, let $\hat{p}^e_{i,t} \triangleq \frac{1}{\abs{\mathbb{T}^{(t)}_e}} \sum_{s \in \mathbb{T}^{(t)}_e} \mathds{1}\{\tilde X(s) = i\}$. Then,
\begin{eqnarray}
\mathbb{E}[\mathcal{{A}}_T(\tilde{{X}}^e_{1,t})] &=&  \mathbb{E} \left[\sup_{\Pi \in \mathcal{P}} \sum_{i=1}^N g(\pi_i) (\hat{p}^e_{i,t} - p_{i,t+T})  \right]\nonumber\\
&=&  \mathbb{E} \left[\sup_{\Pi \in \mathcal{P}} \sum_{i=1}^N g(\pi_i) \left(\hat{p}^e_{i,t} - \frac{1}{\abs{\mathbb{T}^{(t)}_e}} \sum_{s \in \mathbb{T}^{(t)}_e} p_{i,s}
 +   \frac{1}{\abs{\mathbb{T}^{(t)}_e}} \sum_{s \in \mathbb{T}^{(t)}_e} p_{i,s} - p_{i,t+T}\right)\right ] \nonumber \\
&\stackrel{(a)}{\leq}& \mathbb{E} \left[\sup_{\Pi \in \mathcal{P}}\sum_{i=1}^N g(\pi_i)  \left(\hat{p}^e_{i,t} - \frac{1}{\abs{\mathbb{T}^{(t)}_e}} \sum_{s \in \mathbb{T}^{(t)}_e} p_{i,s} \right ) + \Delta_{t,T}^{(e)}\right],
 \label{eq:rademacher1}
\end{eqnarray}
where $\Delta_{t,T}^{(e)} \triangleq \mathbb{E} \sup_{\Pi \in \mathcal{P}} \sum_{i=1}^N g(\pi_i)  d^{(e)}_i(t+T)$, $d^{(e)}_i(t,T) \triangleq \frac{1}{\abs{\mathbb{T}^{(t)}_e}} \sum_{s \in \mathbb{T}^{(t)}_e} \left \vert  p_{i,s} - p_{i,t+T} \right \vert$, and $(a)$ follows from the triangular inequality. Let us consider a sequence of RVs $\bar{{X}}_{1,t}$ independent of $\tilde{{X}}_{1,t}$, but with the same distribution. Thus, $p_{i,s} = \mathbb{E}[\mathds{1}\{\bar{{X}}_{1,t}(s) = i\}]$ $\forall$ $i$, where $\bar{{X}}_{1,t}(s)$ is the $s$th component of $\bar{{X}}_{1,t}$. Substituting the values of $p_{i,s}$ and $\hat{p}_{i,t}^e$, the first term in \eqref{eq:rademacher1} becomes
\begin{eqnarray}
\nonumber \mathbb{E}_{\tilde X}\left[ \sup_{\Pi \in \mathcal{P}}\sum_{i=1}^N g(\pi_i)  \left(\hat{p}^e_{i,t} - \frac{1}{\abs{\mathbb{T}^{(t)}_e}} \sum_{s \in \mathbb{T}^{(t)}_e} p_{i,s} \right)\right] &=&  \mathbb{E}_{\tilde X} \left[\sup_{\Pi \in \mathcal{P}}\sum_{i=1}^N g(\pi_i)  \left(\frac{1}{\abs{\mathbb{T}^{(t)}_e}} \sum_{s \in \mathbb{T}^{(t)}_e}\Delta_E X_{i,s,t}\right) \right] \nonumber \\
&\stackrel{(a)}{\leq}& \mathbb{E}_{\tilde X, \hat X}\left[\sup_{\Pi \in \mathcal{P}}\sum_{i=1}^N g(\pi_i)  \left(\frac{1}{\abs{\mathbb{T}^{(t)}_e}} \sum_{s \in \mathbb{T}^{(t)}_e}\Delta X_{i,s,t}\right)\right] \nonumber \\
\nonumber &\stackrel{(b)}{\leq}& \mathbb{E}_{\tilde X, \hat X, \bm{\sigma}}\left[\sup_{\Pi \in \mathcal{P}}\sum_{i=1}^N g(\pi_i)  \left(\frac{1}{\abs{\mathbb{T}^{(t)}_e}} \sum_{s \in \mathbb{T}^{(t)}_e} \sigma_{i,s} \Delta X_{i,s,t}\right)\right] \\
\nonumber &\leq&  \mathbb{E}_{\tilde X, \bm{\sigma}}\left[\sup_{\Pi \in \mathcal{P}}\sum_{i=1}^N g(\pi_i)  \left(\frac{1}{\abs{\mathbb{T}^{(t)}_e}} \sum_{s \in \mathbb{T}^{(t)}_e} \sigma_{i,s} \mathds{1}\{\tilde X(s) = i\}
\right)\right], \\
\label{eq:app_mainres2proof}
\end{eqnarray}
where $\Delta_E X_{i,s,t} \triangleq  \mathds{1}\{\tilde X(s) = i\} - \mathbb{E} [\mathds{1}\left\{\bar{{X}}_{1,t}(s) = i\right\}]$, and $\Delta X_{i,s,t} \triangleq  \mathds{1}\{\tilde X(s) = i\} - \mathds{1}\left\{\bar{{X}}_{1,t}(s) = i\right\}$. In \eqref{eq:app_mainres2proof}, $(a)$ follows from the convexity property, and $(b)$ follows from the fact that $\Delta X_{i,s,t}$ and $\sigma_{i,s} \Delta X_{i,s,t}$ have the same distribution, where the Rademacher RVs $\sigma_{i,s} \in \{-1,1\}$ are {\iid} with probability $1/2$ each. We also have $\bm{\sigma} \triangleq \{\sigma_{i,s}: 1 \leq i \leq N, s \in \mathbb{T}^{(t)}_e\}$. Using \defref{def:rademachercomplexity}, we have $\mathbb{E}[\mathcal{{A}}_T(\tilde{{X}}^e_{1,t})]  \leq \mathcal{R}^{(t)}_e + \Delta_{t,T}^{(e)}$. Similar analysis holds for the odd term leading to $\mathbb{E}[\mathcal{{A}}_T(\tilde{{X}}^o_{1,t})] \leq \mathcal{R}^{(t)}_o + \Delta_{t,T}^{(o)}$, where $\mathcal{R}^{(t)}_o$ and $\Delta_{t,T}^{(o)}$ are defined similarly to $\mathcal{R}^{(t)}_e$ and $\Delta_{t,T}^{(e)}$, respectively. Using these, we get $\max\left\{\mathbb{E}\{\mathcal{{A}}_T(\tilde{{X}}^e_{1,t})\}, \mathbb{E}\{\mathcal{{A}}_T(\tilde{{X}}^e_{1,t})\}\right\} \leq \max\{\mathcal{R}^{(t)}_e, \mathcal{R}^{(t)}_o\} + \max\{\Delta_{t,T}^{(e)}, \Delta_{t,T}^{(o)}\}$. Finally, note that $t = \sum_{j=1}^{2m} a_i \leq 2m \max_{1\leq i \leq 2m} a_i$, which implies $m \geq \frac{t}{2\max_{1\leq i \leq 2m} a_i}$. Using these results in \thrmref{thm:main_res1}, we get the desired result. This completes the proof of \thrmref{thm:main_res2}.
$\blacksquare$

\bibliographystyle{IEEEtran}
\bibliography{IEEEabrv,caching2016}

\end{document}